\documentclass[conference]{IEEEtran}

\makeatletter
\def\bstctlcite{\@ifnextchar[{\@bstctlcite}{\@bstctlcite[@auxout]}}
\def\@bstctlcite[#1]#2{\@bsphack
  \@for\@citeb:=#2\do{%
    \edef\@citeb{\expandafter\@firstofone\@citeb}%
    \if@filesw\immediate\write\csname #1\endcsname{\string\citation{\@citeb}}\fi}%
  \@esphack}
\makeatother

\usepackage{amsmath}
\usepackage{amsfonts}
\usepackage{amssymb}
\usepackage{amsthm}
\usepackage{tikz}
\usepackage{bm}
\usepackage{makecell}
\usepackage{caption}
\newtheorem{lemma}{Lemma}

\newtheorem{corollary}{Corollary}
\newtheorem{remark}{Remark}

\begin{document}

\title{An Improved Cooperative Repair Scheme for Reed-Solomon Codes}

\author{\IEEEauthorblockN{Yaqian Zhang, ~Zhifang Zhang}\\
\IEEEauthorblockA{\fontsize{9.8}{12}\selectfont KLMM, Academy of Mathematics and Systems Science, Chinese Academy of Sciences, Beijing 100190, China\\
 School of Mathematical Sciences, University of Chinese Academy of Sciences, Beijing 100049, China\\
Emails: zhangyaqian15@mails.ucas.ac.cn, ~zfz@amss.ac.cn}

}

\maketitle

\begin{abstract}
  Dau et al. recently extend Guruswami and Wootters' scheme (STOC'2016) to cooperatively repair two or three erasures in Reed-Solomon (RS) codes.
  However, their scheme restricts to either the case that the characteristic of $F$ divides the extension degree $[F\!:\!B]$ or some special failure patterns, where $F$ is the base field of the RS code and $B$ is the subfield of the repair symbols.
  In this paper, we derive an improved cooperative repair scheme that removes all these restrictions. That is, our scheme applies to any characteristic of $F$ and can repair all failure patterns of two or three erasures.
\end{abstract}

\section{Introduction}
Reed-Solomon (RS) codes, as a kind of maximum distance separable (MDS) codes, are extensively used in distributed storage systems for providing the optimal trade-off between redundancy and reliability. For example, a [14,10] RS code is deployed in the file system of Facebook \cite{Sathiamoorthy13}.
Given a finite field $F$, a file represented by a vector in $F^k$ is encoded into a codeword in $F^n$ using an $[n,k]$ RS code. Each coordinate of the codeword is stored in a storage node.
When some node fails, i.e., the symbol stored in that node is erased, a self-sustaining system should be able to repair the failed node by downloading data from the surviving nodes (\emph{helper nodes}). An important metric for the node repair efficiency is the \emph{repair bandwidth}, namely, the total amount of data communicated during the repair process.

A naive repair method for RS codes requires downloading the whole file to recover just one single node. This is wasteful and the repair cost is far more expensive than the optimal repair bandwidth indicated by the cut-set bound \cite{Dimakis07}.
Guruswami and Wootters \cite{Guruswami16} proposed a linear repair scheme for repairing one erasure in RS codes that can significantly reduce the repair bandwidth. The key idea is to perform the repair over a subfield $B$ of $F$, i.e., computing the erased symbol from enough of its traces from $F$ to $B$.
Later, Tamo et al. \cite{Tamo17} defined a class of RS codes over a sufficiently large field $F$, i.e., $[F:B]\approx n^n$, that permits a repair scheme with bandwidth achieving the cut-set bound.

In some scenarios, however, multiple node failures are quite common. There are two typical models for repairing multiple erasures. One is the centralized repair, where a special node called data center is assumed to generate all the replacement nodes. RS codes with centralized repair are considered in \cite{Tamo18, Bartan17}. The other is the cooperative repair \cite{Hu2010}, where the replacement nodes are generated in a distributed and cooperative way. It was proved in \cite{Ye2018} that an MDS code achieving the optimal bandwidth in the cooperative repair mode naturally attains the optimal bandwidth in the centralized repair mode. Moreover, due to the distributed pattern, the cooperative repair model fits the system better than the centralized one. In this paper, we focus on the cooperative repair of RS codes in the case of two or three erasures.

Recently, Dau et al. \cite{Dau18} extend Guruswami and Wootters' scheme \cite{Guruswami16} to cooperatively repair RS codes with two or three erasures. However, for two erasures, their scheme is either restricted to the case that the characteristic of $F$ divides the extension degree $[F:B]$, or requires a sequential recovery of the two erasures. And for three erasures, their scheme only applies to some special failure patterns. In this paper, we develop an improved cooperative repair scheme for RS codes that removes all restrictions when repairing two erasures and can apply to all failure patterns of three erasures.

The remaining of the paper is organized as follows. First, we recall some necessary preliminaries in Section \ref{sec2}. Then our main results -- the repair schemes for two erasures and for three erasures are presented in Section \ref{sec3} and Section \ref{sec4} respectively. Finally, Section \ref{sec5} concludes the paper.

\section{Preliminaries}\label{sec2}
\subsection{Notations and definitions}
Throughout the paper, we use $[n]$ to denote  $\{1,2,...,n\}$. Let $F$ be a finite field and $B$ be a subfield of $F$ with $[F:B]=t$. The elements in $F$ are called symbols and elements in $B$ are called subsymbols.

Let $A=\{\alpha_1,...,\alpha_n\}$ be a set of distinct elements in $F$. An $[n,k]$ RS code with the {\it evaluator set} $A$, denoted by  $\mathrm{RS_A}$, is  defined as
\begin{equation*}\label{RS}
\mathrm{RS_A}=\{(f(\alpha_1),...,f(\alpha_n)):f\in F[x], \mathrm{deg}(f)< k\}.
\end{equation*}

For any $x\in F$, its trace from $F$ to $B$ is defined as
$$
\mathrm{Tr}_{F/B}(x)=x+x^{|B|}+\cdots+x^{|B|^{t-1}}
$$
which is always a subsymbol in $B$. For simplicity, we use $\rm{Tr}$ instead of $\mathrm{Tr}_{F/B}$ to denote the trace function from $F$ to $B$ when the two fields are clear from the context.

\subsection{Guruswami and Wootters' scheme}\label{sec2-2}
Guruswami and Wootters' scheme \cite{Guruswami16} for repairing one erasure in $\mathrm{RS_A}$ relies on two basic facts.

One is that the dual of an RS code is still a generalized RS code, i.e.,
$$
\mathrm{RS_A^\bot}=\{(\lambda_1g(\alpha_1),...,\lambda_ng(\alpha_n)):g\!\in\! F[x], \mathrm{deg}(g)< n\!-\!k\},
$$
where $\lambda_i$, $i\!\in\![n]$, are called {\it multipliers} which are nonzero elements in $F$ and determined from $A$. Hereafter, a polynomial of degree less than $n-k$ is called a check polynomial of the $[n,k]$ RS code.

The other fact is that every element in $F$ can be computed  from its $t$ independent traces as illustrated in the following lemma.
\begin{lemma}\cite{Guruswami16}\label{lemma1}
Suppose $\{\zeta_1,...,\zeta_t\}$ is a basis of $F$ over $B$. Then for every $\gamma\in F$, $\gamma$ can be recovered from the $t$ subsymbols $\{\rm{Tr}(\zeta_{1}\gamma),...,\rm{Tr}(\zeta_{t}\gamma)\}$, i.e.,  $\gamma=\sum_{i=1}^{t}\rm{Tr}(\zeta_{i}\gamma)\zeta_{i}^{\bot}$, where $\{\zeta_{i}^{\bot}\}_{i\in[t]}$ is the dual basis of $\{\zeta_{i}\}_{i\in[t]}$.
\end{lemma}
Now suppose $n-k\geq|B|^{t-1}$ and one symbol, say $f(\alpha^*)$ for some $\alpha^*\in A$, is erased. To recover $f(\alpha^*)$, first define $t$ polynomials $p_{i}(x)\!= \!\mathrm{Tr}(u_i(x-\alpha^*))/(x-\alpha^*)$, $i
\!\in\![t]$, where $\{u_i\}_{i\in[t]}$ is a basis of $F$ over $B$. It can be seen that ${\rm deg}(p_{i}(x))=|B|^{t-1}-1<n-k$, therefore these $p_i(x)$'s are check polynomials for $\mathrm{RS_A}$, defining
$t$ parity check equations: $\sum_{\alpha\in A}\lambda_{\alpha}p_{i}(\alpha)f(\alpha)=0$, $i\in[t]$. Then applying the trace function and using the $B$-linearity of the traces, it follows that for $i\in[t]$,
\begin{equation}\label{3}
\mathrm{Tr}(\lambda_{\alpha^*}p_{i}(\alpha^*)f(\alpha^*))=-\sum_{\alpha\in A\setminus \{\alpha^*\}}\mathrm{Tr}(\lambda_{\alpha}p_{i}(\alpha)f(\alpha)).
\end{equation}

By the definition of $p_i(x)$ one can see $p_i(\alpha^*)=u_i$ for $i\in[t]$, thus $\{p_i(\alpha^*)\}_{i\in[t]}$ forms a basis of $F$ over $B$. From Lemma \ref{lemma1} we know the left sides of the equations in (\ref{3}), i.e., $\{\mathrm{Tr}(\lambda_{\alpha^*}p_i(\alpha^*)f(\alpha^*))\}_{i\in[t]}$, suffices to recover $\lambda_{\alpha^*}f(\alpha^*)$, and thus the symbol $f(\alpha^*)$ since $\lambda_{\alpha^*}\neq0$.

Therefore, it is left to collect the right sides of the  equations in (\ref{3}) from the helper nodes.
By the definition of $p_i(x)$, $\mathrm{Tr}(\lambda_{\alpha}p_{i}(\alpha)f(\alpha))=\mathrm{Tr}(u_{i}(\alpha-\alpha^*))\mathrm{Tr}(\frac{\lambda_{\alpha}f(\alpha)}{\alpha-\alpha^*})$. Obviously, it is sufficient to download $\mathrm{Tr}(\frac{\lambda_{\alpha}f(\alpha)}{\alpha-\alpha^*})$ from the node storing $f(\alpha)$ for all $\alpha\in A\setminus\{\alpha^*\}$. So the total repair bandwidth is $n-1$ subsymbols in $B$.

\subsection{The cooperative repair model}
We simply recall the cooperative repair model introduced in \cite{Hu2010}. Suppose $r$ replacement nodes are to be generated to replace  $r$ failed nodes respectively. The process is accomplished in two phases:

In Phase 1, each replacement node connects to $d\ (\leq n-r)$ helper nodes and downloads $\beta_1$ subsymbols from each.

In Phase 2, the $r$ replacement nodes exchange $\beta_2$ subsymbols with each other.

Note that Phase 2 may be accomplished in multiple rounds. Here we assume synchronized and simultaneous channel, namely, all $r$ nodes can send data to others simultaneously in one round.
In the following, an $m$-round repair means a cooperative repair scheme that requires $m$-round comminication in Phase 2. Obviously, one-round repair is preferred with respect to the round complexity.


\section{Cooperative Repair of Two Erasures}\label{sec3}
A one-round repair scheme for two erasures in $\mathrm{RS_A}$ is designed in this section. As in \cite{Guruswami16,Dau18},  it assumes $n-k\geq|B|^{t-1}$.

Denote $K=\{x\in F:{\rm Tr}(x)=0\}$, and for any $\alpha,\beta\in F$ with $\alpha\neq\beta$, define
$$
K_{\alpha,\beta}=\{x\in F: \mathrm{Tr}((\alpha-\beta)x)=0\}.
$$
Obviously, $K_{\alpha,\beta}=\frac{K}{\alpha-\beta}$. Moreover, $K$ and $K_{\alpha,\beta}$ are both $(t-1)$-dimensional subspaces of $F$, and $K_{\alpha,\beta}=K_{\beta,\alpha}$.

WLOG, suppose the two symbols $f(\alpha_1)$ and $f(\alpha_2)$ are erased. we call the two replacement nodes that recover $f(\alpha_1)$ and $f(\alpha_2)$ respectively as node 1 and node 2 . For simplicity, hereafter we use the notation $K_{1,2}$ instead of $K_{\alpha_1,\alpha_2}$.

Let $\{u_1,..,u_{t-1}\}$ be a basis of $K_{1,2}$ over $B$. Choose $\delta\in F$ with $\mathrm{Tr}(\delta)=1$ and set $u_t=\frac{\delta}{\alpha_1-\alpha_2}$. Then $\{u_1,...,u_t\}$ form a basis of $F$ over $B$ because $u_t\in F\setminus K_{1,2}$. Choose a nonzero element $\gamma\in K$ and define $2t$ check polynomials:
$$
p_{i}(x)=\frac{\mathrm{Tr}(u_i(x-\alpha_1))}{x-\alpha_1}, \ \
q_{i}(x)=\frac{ \gamma \mathrm{Tr}(u_i(x-\alpha_2))}{x-\alpha_2}, \ \  i\in[t].
$$
It is easy to see that
for $i\in[t-1]$, $p_{i}(\alpha_1)=u_i$, $p_{i}(\alpha_2)=0$ because $u_i\in K_{1,2}$, and $p_{t}(\alpha_1)=u_t$, $p_{t}(\alpha_2)=\frac{1}{\alpha_1-\alpha_2}$. Similarly, $q_{i}(\alpha_2)=\gamma u_i$, $q_{i}(\alpha_1)=0$ for $i\in[t-1]$, and $q_{t}(\alpha_2)=\gamma u_t$, $q_{t}(\alpha_1)=\frac{ \gamma }{\alpha_1-\alpha_2}$.

In Phase 1,  node 1 and node 2 each downloads one subsymbol from each of the $n-2$ surviving nodes. Specifically, node 1 uses the $t$ check polynomials $p_{i}(x)$ to get $t$ check equations \footnote{For simplicity, we omit the multipliers $\lambda_i$ (or assume $\lambda_i=1$), $i\in[n]$, in the check equations because they are explicitly determined from $A$ and has no influence on the repair property.}:
\begin{equation*}
\mathrm{Tr}\big(p_{i}(\alpha_1)f(\alpha_1)\big)\!+\!\mathrm{Tr}\big(p_{i}(\alpha_2)f(\alpha_2)\big)
\!=\!-\!\sum_{\alpha\in A_{\!-\!1,\!-\!2\!}}\!\!\mathrm{Tr}\big(p_{i}(\alpha)f(\alpha)\big)
\end{equation*}
for $i\in [t]$, where $A_{-1,-2}=A\setminus\{\alpha_1,\alpha_2\}$.
Similarly, node 2 uses $q_{i}(x),i\in [t]$ to create $t$ check equations:
\begin{equation*}
\mathrm{Tr}\big(q_{i}(\alpha_1)f(\alpha_1)\big)\!+\!\mathrm{Tr}\big(q_{i}(\alpha_2)f(\alpha_2)\big)
\!=\!-\!\sum_{\alpha\in A_{\!-\!1,\!-\!2\!}}\!\!\mathrm{Tr}\big(q_{i}(\alpha)f(\alpha)\big).
\end{equation*}

As introduced in Section \ref{sec2-2}, node 1 and node 2 can derive the left sides of the equations by collecting the data needed to compute the right sides. The details are displayed in Table \ref{t1}.

\begin{table}[ht]
	\centering
\begin{tabular}{|c|c|c|}
	\hline & \makecell*[l]{download from\\ $f(\alpha),~\alpha\in A_{\!-\!1,\!-\!2\!}$} & obtain\\\hline
node 1 & $\mathrm{Tr}(\frac{f(\alpha)}{\alpha-\alpha_1})$ & $\begin{cases}
	\mathrm{Tr}(u_{j}f(\alpha_1)), \  j\in[t-1] \\
	\mathrm{Tr}(u_{t}f(\alpha_1))\!+\!\mathrm{Tr}(\frac{f(\alpha_2)}{\alpha_1-\alpha_2})
\end{cases}$\\\hline
node 2 & $\mathrm{Tr}(\frac{\gamma f(\alpha)}{\alpha-\alpha_2})$ &
$\begin{cases}
	\mathrm{Tr}(\gamma u_{j}f(\alpha_2)), \  j\in[t-1] \\
	\mathrm{Tr}(\gamma u_{t}f(\alpha_2))\!+\!\mathrm{Tr}(\frac{\gamma f(\alpha_1)}{\alpha_1-\alpha_2})
\end{cases}$\\\hline
\end{tabular}
\caption{Phase 1 of the cooperative repair for two erasures.}\label{t1}
\end{table}

That is, after Phase 1 each node (eg. node 1) obtains $t-1$ independent traces (eg. $\{\mathrm{Tr}(u_{j}f(\alpha_1))\}_{j\in[t-1]}$) and one mixed term (eg. $\mathrm{Tr}(u_{t}f(\alpha_1))\!+\!\mathrm{Tr}(\frac{f(\alpha_2)}{\alpha_1-\alpha_2})$). According to Lemma \ref{lemma1}, it needs one more independent trace to recover the erased symbol. We next show this can be accomplished in Phase 2 by exchanging one subsymbol between the two nodes.

 Since $\gamma\in K$, then $\frac{\gamma}{\alpha_1-\alpha_2}\in K_{1,2}={\rm Span}\{u_j:j\in[t-1]\}$. Thus $\mathrm{Tr}(\frac{\gamma f(\alpha_1)}{\alpha_1-\alpha_2})$ can be generated as a $B$-linear combination of $\{\mathrm{Tr}(u_{j}f(\alpha_1))\}_{j\in[t-1]}$ which have been obtained by node 1 in Phase 1. Therefore, node 1
 directly sends $\mathrm{Tr}(\frac{\gamma f(\alpha_1)}{\alpha_1-\alpha_2})$ to node 2 in Phase 2. It is easy to see that this transition can make node 2 finish the recovery.

On the other hand, suppose
\begin{equation}\label{key1}
\frac{1}{\alpha_1-\alpha_2}=\sum_{i=1}^{t}a_i\gamma u_{i}, \ \ a_i\in B.
\end{equation}
Then using $\{a_i\}_{i\in[t]}$ as linear combination coefficients of the $t$ terms node 2 obtains in Phase 1, then it can derive
$\mathrm{Tr}(\frac{f(\alpha_2)}{\alpha_1-\alpha_2})+a_{t}\mathrm{Tr}(\frac{\gamma f(\alpha_1)}{\alpha_1-\alpha_2})$ which is exactly the subsymbol that node 2 sends to node 1 in Phase 2. We will show this transition also makes node 1 finish its recovery.

Because subtracting this subsymbol from the mixed term obtained in Phase 1, node 1 gets $\mathrm{Tr}(u_{t}f(\alpha_1))\!-\!a_{t}\mathrm{Tr}(\frac{\gamma f(\alpha_1)}{\alpha_1-\alpha_2})=\mathrm{Tr}(\frac{\delta-a_{t}\gamma}{\alpha_1-\alpha_2}f(\alpha_1))$.
However,
\begin{equation}\label{key2}
\mathrm{dim}_{B}\{u_1,...,u_{t-1},\frac{\delta-a_{t}\gamma}{\alpha_1-\alpha_2}\}=t,
\end{equation}
which follows from $\delta-a_{t}\gamma\notin K$ and thus $\frac{\delta-a_{t}\gamma}{\alpha_1-\alpha_2}\notin K_{1,2}$. As a result, node 1 can recover $f(\alpha_1)$ by Lemmma \ref{lemma1}.

\begin{remark}\label{remark1}
Comparing with the one-round repair scheme given in \cite{Dau18}, we modify the check polynomials $q_i(x)$, $i\!\in\![t]$, by multiplying a parameter $\gamma\!\in\! K$. Consequently, in Phase 2 node 2 can eliminate the interference $\mathrm{Tr}(\frac{\gamma f(\alpha_1)}{\alpha_1-\alpha_2})$ with the help of node 1 because $\frac{\gamma}{\alpha_1-\alpha_2}\!\in\! K_{1,2}$, while the scheme in \cite{Dau18} realizes this elimination by assuming $\frac{1}{\alpha_1-\alpha_2}\!\in\! K_{1,2}$, thus the restriction ${\rm char}(F)\mid t$ is needed there.
Meanwhile, we further use the facts (\ref{key1}) and (\ref{key2}) to prove that node 1 can also accomplish the recovery after simultaneous exchange of the proper data in Phase 2. However, these two facts were not noticed in \cite{Dau18}, thus it only derives a two-round repair scheme even after the restriction ${\rm char}(F)\mid t$ is removed.
\end{remark}

\section{Cooperative Repair of Three Erasures}\label{sec4}

Now we discuss the cooperative repair of three erasures in $\rm RS_A$. As before, we require $n-k\geq |B|^{t-1}$.

For three distinct elements $\alpha_1, \alpha_2, \alpha_3\in F$, define
$$
K_{1,2,3}=\{x\in F: \mathrm{Tr}(\alpha_1x)=\mathrm{Tr}(\alpha_2x)=\mathrm{Tr}(\alpha_3x)\}.
$$
Then $K_{1,2,3}$ is a linear space over $B$, and it is easy to see $K_{1,2,3}\!=\!K_{1,2}\cap K_{2,3}\cap K_{1,3}$.
Actually, $K_{1,2,3}$ is the intersection of any two of the three spaces because $K_{i,j}\!\cap\! K_{i,k}\!\subseteq \! K_{j,k}$ for any $\{i,j,k\}\!=\![3]$.

WLOG, suppose the three symbols $f(\alpha_1), f(\alpha_2), f(\alpha_3)$ are erased and three replacement nodes called node 1, node 2, node 3 are to recover them respectively. Since our repair scheme depends on the dimension of the space $K_{1,2,3}$, we first derive two corollaries to describe the dimension.

\begin{lemma}\label{lemma0}
Suppose $K=\mathrm{Ker(Tr)}$, and $S$ is a subspace of $F$ with dimension $s$, then
$s-1\leq \mathrm{dim}_{B}(S\cap K)\leq s$.
\end{lemma}
\begin{proof}
It is obvious that dim$_{B}(S\cap K)\leq s$. Let $\mathrm{Tr}|_{S}$ denote the trace function from $F$ to $B$ restricted to the subspace $S$. Then $\mathrm{Tr}|_{S}$ is a linear map from $S$ to $B$. Since $\mathrm{dim}_{B}\mathrm{Im}(\mathrm{Tr}|_{S})\leq1$, then $\mathrm{dim}_{B}\mathrm{Ker}(\mathrm{Tr}|_{S})$=$\mathrm{dim}_{B}(S\cap K)\geq s-1$.
\end{proof}
\begin{corollary}
$t-2\leq\mathrm{dim}_{B}K_{1,2,3}\leq t-1$.
\end{corollary}
The corollary follows from Lemma \ref{lemma0} and the facts $K_{1,2,3}=K_{1,2}\cap K_{2,3}$ and $|K_{1,2,3}|=|K_{1,2}\cap K_{2,3}|=|\frac{\alpha_2-\alpha_3}{\alpha_1-\alpha_2}K\cap K|$. Note that this corollary is also displayed in \cite{Dau18}.

\begin{lemma}\label{lemma3}
For any $\sigma\!\in\! F$, $\sigma K\!=\!K$ iff $\sigma\!\in \!B^*$, where $K\!=\!\mathrm{Ker(Tr)}$ and $B^*$ denotes the set of nonzero elements in $B$.
\end{lemma}
\begin{proof}
If $\sigma\in B^*$, it is obvious that $\sigma K=K$. Conversely, suppose $\sigma K=K$. It is evident that $\sigma\neq 0$. Let $\{ z_1,...,z_{t-1}\}$ be a basis of $K$ over $B$, then for $j\in[t-1]$, $\sigma z_j\in K$. We extend $\{z_j\}_{j\in[t-1]}$ to a basis of $F$ over $B$, denoted by $Z=\{ z_1,...,z_{t}\}$. Let $Z^\bot=\{ z_1^\bot,...,z_{t}^\bot\}$ be the dual basis of $Z$. We claim that $z_{t}^\bot=(\mathrm{Tr}(z_{t}))^{-1}$ since $\mathrm{Tr}(z_i(\mathrm{Tr}(z_{t}))^{-1})=0, i\in[t-1]$, and $\mathrm{Tr}(z_t(\mathrm{Tr}(z_{t}))^{-1})=1$. By the uniqueness of dual basis, it follows $z_{t}^\bot=(\mathrm{Tr}(z_{t}))^{-1}$. Now let $\sigma=\sum_{i=1}^{t}a_iz_i^\bot$, with $a_i\in B$. From Lemma \ref{lemma1}, we know $a_i=\mathrm{Tr}(\sigma z_i), i\in[t]$. Thus $a_j=0$ for $j\in[t-1]$ because $\sigma z_j\in K$, while $a_t=\mathrm{Tr}(\sigma z_t)\neq0$. Therefore, $\sigma=a_tz_t^\bot=\mathrm{Tr}(\sigma z_t)(\mathrm{Tr}(z_{t}))^{-1}\in B^*$.
\end{proof}

\begin{corollary}\label{corollary1}
 $\mathrm{dim}_{B}K_{1,2,3}=t-1$ iff $\frac{\alpha_i-\alpha_j}{\alpha_{i'}-\alpha_{j'}}\in B^*$, for any $i,j,i',j'\in\{1,2,3\}$ with $i\neq j$, $i'\neq j'$.
\end{corollary}

\begin{proof}
One can verify that ${\rm dim}_{B}K_{1,2,3}\!=\!t\!-\!1\iff K_{1,2}\!=\!K_{2,3}\!=\!K_{1,3}\iff \frac{K}{\alpha_1-\alpha_2}\!=\!\frac{K}{\alpha_2-\alpha_3}\!=\!\frac{K}{\alpha_3-\alpha_1}\iff K\!=\!\frac{\alpha_i-\alpha_j}{\alpha_{i'}-\alpha_{j'}}K$ for any $i\!\neq\! j, i'\!\neq\! j'$ and $i,j,i',j'\!\in\!\{1,2,3\}$. By Lemma \ref{lemma3}, the corollary follows immediately.
\end{proof}


For simplicity, denote $l=\mathrm{dim}_{B}K_{1,2,3}$, then $l=t-1$ or $t-2$. Let $U=\{u_i\}_{i\in[l]}$, $V=\{v_i\}_{i\in[l]}$, $W=\{w_i\}_{i\in[l]}$ be three bases of $K_{1,2,3}$ over $B$ which can be the same basis. We extend $U$ to a basis of $K_{1,2}$, denoted by $\{u_i\}_{i\in[t-1]}$. Similarly, extend $V$ to a basis of $K_{2,3}$, denoted by $\{v_i\}_{i\in[t-1]}$ and extend $W$ to a basis of $K_{1,3}$, denoted by $\{w_i\}_{i\in[t-1]}$. Then we further extend them to three bases of $F$ over $B$, i.e., $U'=\{u_i\}_{i\in[t]}$, $V'=\{v_i\}_{i\in[t]}$, $W'=\{w_i\}_{i\in[t]}$. Next define $3t$ check polynomials, i.e., for $1\leq i\leq t$,
$$\begin{aligned}
p_{i}(x)&=\frac{ \mathrm{Tr}(u_i(x-\alpha_1))}{x-\alpha_1}, \ \ \ \ \
q_{i}(x)=\frac{\gamma_1 \mathrm{Tr}(v_i(x-\alpha_2))}{x-\alpha_2},\\
r_{i}(x)&=\frac{ \gamma_2 \mathrm{Tr}(w_i(x-\alpha_3))}{x-\alpha_3},
\end{aligned}$$
where $\gamma_1, \gamma_2$ are two nonzero elements in $F$. In the following, we illustrate the repair schemes in Section \ref{sub1} and  \ref{sub2} for $l=t-1$ and $l=t-2$ respectively.

\subsection{$l=t-1$}\label{sub1}
A one-round repair scheme is designed in this case.

Since
$K_{1,2,3}=K_{1,2}=K_{2,3}=K_{1,3}$ when $l=t-1$, we can set $u_t=v_t=w_t=\frac{\delta}{\alpha_1-\alpha_2}$, where $\delta$ is chosen from $F$ with $\mathrm{Tr}(\delta)=1$. By the definition of $p_i(x)$, $q_i(x)$ and $r_i(x)$, it has
\begin{equation*}
\begin{cases}
p_{j}(\alpha_1)=u_j,\ \ \ p_{j}(\alpha_2)=p_{j}(\alpha_3)=0, \ \ \ \ j\in[t-1],  \\
p_{t}(\alpha_1)=u_t, \ \ \ p_{t}(\alpha_2)=p_{t}(\alpha_3)=\frac{1}{\alpha_1-\alpha_2}; \\
\end{cases}
\end{equation*}
\begin{equation*}
\begin{cases}
q_{j}(\alpha_2)=\gamma_1 v_j,\ \  q_{j}(\alpha_1)=q_{j}(\alpha_3)=0,  \ \ \ \  j\in[t-1],  \\
q_{t}(\alpha_2)=\gamma_1 v_t, \ \  q_{t}(\alpha_1)=q_{t}(\alpha_3)=\frac{ \gamma_1 }{\alpha_1-\alpha_2};\\
\end{cases}
\end{equation*}
\begin{equation*}
\begin{cases}
r_{j}(\alpha_3)=\gamma_2 w_j,\ \  r_{j}(\alpha_1)=r_{j}(\alpha_2)=0, \ \ \ \  j\in[t-1],  \\
r_{t}(\alpha_3)=\gamma_2 w_t, \ \  r_{t}(\alpha_1)=r_{t}(\alpha_2)=\frac{ \gamma_2 }{\alpha_1-\alpha_2}.
\end{cases}
\end{equation*}
Note here
$p_{t}(\alpha_3)=\frac{1}{\alpha_3-\alpha_1}\mathrm{Tr}(\frac{(\alpha_3-\alpha_1)\delta}{\alpha_1-\alpha_2})
=\frac{1}{\alpha_1-\alpha_2}\mathrm{Tr}(\delta)
=\frac{1}{\alpha_1-\alpha_2}$,
where the second equality holds because it has  $\frac{\alpha_3-\alpha_1}{\alpha_1-\alpha_2}\in B^*$ from Corollary \ref{corollary1}. Other details of the computations are evident.

In Phase 1, the three nodes obtain some independent traces and mixed terms according to the check equations defined by $p_i(x),q_i(x), r_i(x)$ respectively, $i\in[t]$. As in Section \ref{sec3} we illustrate this process in Table \ref{t2}.
\begin{table}[ht]
	\centering
\resizebox{\columnwidth}{!}{
\begin{tabular}{|c|c|c|}
	\hline & \makecell[l]{download from $f(\alpha),$\\ $\alpha\in A_{\!-\!1,\!-\!2,\!-\!3}$} & obtain\\\hline
node 1 & $\mathrm{Tr}(\frac{f(\alpha)}{\alpha-\alpha_1})$ & $\begin{cases}
 \mathrm{Tr}(u_{j}f(\alpha_1)), ~~~~~j\in[t-1] \\
 \mathrm{Tr}(u_{t}f(\alpha_1))\!+\!\mathrm{Tr}(\frac{f(\alpha_2)}{\alpha_1-\alpha_2})\!+\!\mathrm{Tr}(\frac{f(\alpha_3)}{\alpha_1-\alpha_2})
\end{cases}$\\\hline
node 2 & $\mathrm{Tr}(\frac{\gamma_1 f(\alpha)}{\alpha-\alpha_2})$ &
$\begin{cases}
 \mathrm{Tr}(\gamma_1v_{j}f(\alpha_2)), ~~~~~ j\in[t-1] \\
 \mathrm{Tr}(\gamma_1v_{t}f(\alpha_2))\!+\!\mathrm{Tr}(\frac{\gamma_1f(\alpha_1)}{\alpha_1-\alpha_2})\!+\!\mathrm{Tr}(\frac{\gamma_1f(\alpha_3)}{\alpha_1-\alpha_2})
\end{cases}$\\\hline
node 3&$\mathrm{Tr}(\frac{\gamma_2 f(\alpha)}{\alpha-\alpha_3})$&$\begin{cases}
 \mathrm{Tr}(\gamma_2w_{j}f(\alpha_3)),  ~~~~~ j\in[t-1] \\
 \mathrm{Tr}(\gamma_2w_{t}f(\alpha_3))\!+\!\mathrm{Tr}(\frac{\gamma_2f(\alpha_1)}{\alpha_1-\alpha_2})\!+\!\mathrm{Tr}(\frac{\gamma_2f(\alpha_2)}{\alpha_1-\alpha_2})
\end{cases}$\\\hline
\end{tabular}}
\caption{Phase 1 of the cooperative repair for three erasures ($l=t-1$).}\label{t2}
\end{table}

Next we will show that by properly choosing the parameters $\gamma_1$ and $\gamma_2$, the recovery can be achieved in Phase 2 through one-round communication where each node sends a subsymbol to the other two nodes.

\begin{lemma}\label{lemma4}
Given $\gamma_1,\gamma_2\in F^*$, if the following equations on $x_i,y_i$, $i\in[3]$, have solutions in $B$:
\begin{equation}\label{eq1}
\begin{cases}
x_1+\mathrm{Tr}(\gamma_1^{-1}\gamma_2)y_1=\mathrm{Tr}(\gamma_1^{-1}),\\
\mathrm{Tr}(\gamma_1\gamma_2^{-1})x_1+y_1=\mathrm{Tr}(\gamma_2^{-1}),\\
\mathrm{Tr}(\gamma_1)x_1+\mathrm{Tr}(\gamma_2)y_1\neq1.
\end{cases}
\end{equation}
\begin{equation}\label{eq2}
\begin{cases}
x_2+\mathrm{Tr}(\gamma_2)y_2=\mathrm{Tr}(\gamma_1),\\
\mathrm{Tr}(\gamma_2^{-1})x_2+y_2=\mathrm{Tr}(\gamma_1\gamma_2^{-1}),\\
\mathrm{Tr}(\gamma_1^{-1})x_2+\mathrm{Tr}(\gamma_1^{-1}\gamma_2)y_2\neq1.
\end{cases}
\end{equation}
\begin{equation}\label{eq3}
\begin{cases}
x_3+\mathrm{Tr}(\gamma_1)y_3=\mathrm{Tr}(\gamma_2),\\
\mathrm{Tr}(\gamma_1^{-1})x_3+y_3=\mathrm{Tr}(\gamma_1^{-1}\gamma_2),\\
\mathrm{Tr}(\gamma_2^{-1})x_3+\mathrm{Tr}(\gamma_1\gamma_2^{-1})y_3\neq1.
\end{cases}
\end{equation}
then the recovery can be accomplished in Phase 2 by each replacement node exchanging one subsymbol with the other two nodes in one round.
\end{lemma}
\begin{proof}
First, consider the repair of node 1. Suppose the equation (\ref{eq1}) has a solution $(a_t,b_t)\in B^2$, then we have $\mathrm{Tr}(\gamma_1^{-1}-a_t\delta-b_t\gamma_1^{-1}\gamma_2)=\mathrm{Tr}(\gamma_2^{-1}-b_t\delta-a_t\gamma_1\gamma_2^{-1})=0$ where we use the fact that $\mathrm{Tr}(\delta)=1$. As a result,
\begin{equation*}
\begin{cases}
\frac{\gamma_1^{-1}-a_t\delta-b_t\gamma_1^{-1}\gamma_2}{\alpha_1-\alpha_2}\in K_{1,2}\\
\frac{\gamma_2^{-1}-b_t\delta-a_t\gamma_1\gamma_2^{-1}}{\alpha_1-\alpha_2}\in K_{1,2}
\end{cases}\;.
\end{equation*}
Since $\{v_i\}_{i\in[t-1]}$ and $\{w_i\}_{i\in[t-1]}$ are two bases of $K_{1,2}$, let
\begin{equation}\label{7-}
\begin{cases}
\frac{\gamma_1^{-1}-a_t\delta-b_t\gamma_1^{-1}\gamma_2}{\alpha_1-\alpha_2}=\sum_{i=1}^{t-1}a_iv_{i}\\
\frac{\gamma_2^{-1}-b_t\delta-a_t\gamma_1\gamma_2^{-1}}{\alpha_1-\alpha_2}=\sum_{i=1}^{t-1}b_iw_{i}
\end{cases}
\end{equation}
where $a_i,b_i\in B$ for $i\in[t-1]$. Recall that $v_t=w_t=\frac{\delta}{\alpha_1-\alpha_2}$, combining with (\ref{7-}) we have
\begin{equation}\label{eq4}
\begin{cases}
\frac{1}{\alpha_1-\alpha_2}=\sum_{i=1}^{t}a_i\gamma_1v_{i}+b_t\frac{\gamma_2}{\alpha_1-\alpha_2}\\
\frac{1}{\alpha_1-\alpha_2}=\sum_{i=1}^{t}b_i\gamma_2w_{i}+a_t\frac{\gamma_1}{\alpha_1-\alpha_2}
\end{cases}\;.
\end{equation}
Using $\{a_i\}_{i\in[t]}$ and $\{b_i\}_{i\in[t]}$ as linear coefficients, node 2 and node 3  transmit the following two subsymbols to node 1 respectively:
\begin{equation}\label{eq5}
\begin{cases}
 \sum_{i=1}^{t}a_i\mathrm{Tr}(\gamma_1v_{i}f(\alpha_2))+a_t\mathrm{Tr}(\frac{\gamma_1f(\alpha_1)}{\alpha_1-\alpha_2})+a_t\mathrm{Tr}(\frac{\gamma_1f(\alpha_3)}{\alpha_1-\alpha_2}) \\
 \sum_{i=1}^{t}b_i\mathrm{Tr}(\gamma_2w_{i}f(\alpha_3))+b_t\mathrm{Tr}(\frac{\gamma_2f(\alpha_1)}{\alpha_1-\alpha_2})+b_t\mathrm{Tr}(\frac{\gamma_2f(\alpha_2)}{\alpha_1-\alpha_2})
\end{cases}.
\end{equation}
Then, node 1 subtracts the sum of the two subsymbols in (\ref{eq5}) from the mixed term it obtained  in Phase 1. Using the equalities in (\ref{eq4}) it gets $\mathrm{Tr}(u_{t}f(\alpha_1))\!-\!a_t\mathrm{Tr}(\frac{\gamma_1f(\alpha_1)}{\alpha_1-\alpha_2})\!-\!b_t\mathrm{Tr}(\frac{\gamma_2f(\alpha_1)}{\alpha_1-\alpha_2})
=\mathrm{Tr}(\frac{\delta-a_t\gamma_1-b_t\gamma_2}{\alpha_1-\alpha_2}f(\alpha_1)).$
In order to recover $f(\alpha_1)$, it is sufficient to show that $\{u_1,...,u_{t-1},\frac{\delta-a_t\gamma_1-b_t\gamma_2}{\alpha_1-\alpha_2}\}$ forms a basis of $F$ over $B$, or equivalently,  $\frac{\delta-a_t\gamma_1-b_t\gamma_2}{\alpha_1-\alpha_2}\notin K_{1,2}$. This follows from the third equation in (\ref{eq1}), i.e., $a_t\mathrm{Tr}(\gamma_1)+b_t\mathrm{Tr}(\gamma_2)\neq1$.

In a similar way, node 2 and node 3 can be recovered by considering equations (\ref{eq2}) and (\ref{eq3}) respectively. Due to space limitations we omit the details here.
\end{proof}

Now we are left to specify $\gamma_1$ and $\gamma_2$ such that (\ref{eq1}-\ref{eq3}) have solutions in $B$. There are three cases:

\subsubsection{\bm{$t\geq3$}}

Choose $\gamma_1\in K^*$, and $\gamma_2\in K^*\cap\gamma_{1}K$. We can do this because ${\rm dim}_B(K)=t-1$ and dim$_{B}(K\cap\gamma_{1}K)\geq t-2\geq 1$ from Lemma \ref{lemma0}. Then we have $\mathrm{Tr}(\gamma_1)=\mathrm{Tr}(\gamma_2)=\mathrm{Tr}(\gamma_1^{-1}\gamma_2)=0$. It is easy to verify that (\ref{eq1}-\ref{eq3}) are solvable in $B$.

\subsubsection{\bm{$t=2$} \textbf{and} \bm{${\rm char}(F)\neq3$}}

Choose $\gamma_1=\gamma_2\in K^*$. Then $\mathrm{Tr}(\gamma_1)=\mathrm{Tr}(\gamma_2)=0$, $\mathrm{Tr}(\gamma_1^{-1}\gamma_2)=\mathrm{Tr}(\gamma_1\gamma_2^{-1})=\mathrm{Tr}(1)=2$. It is easy to verify that (\ref{eq1}-\ref{eq3}) are solvable in $B$.

\subsubsection{\bm{$t=2$} \textbf{and} \bm{${\rm char}(F)=3$}}

Choose $\gamma_1=\gamma_2=1$. Since $\mathrm{Tr}(2)=1$, we can set $\delta=2$ in particular, then $u_t=v_t=w_t=\frac{2}{\alpha_1-\alpha_2}$. In this case, we give a straightforward way to complete the exchange phase without concerning the equations (\ref{eq1}-\ref{eq3}). Specifically, the three nodes directly exchange the mixed terms obtained in Phase 1 with each other. Then all of the three nodes can get:
\begin{eqnarray*}
&&\begin{pmatrix}
\mathrm{Tr}(u_{t}f(\alpha_1))+\mathrm{Tr}(\frac{f(\alpha_2)}{\alpha_1-\alpha_2})+\mathrm{Tr}(\frac{f(\alpha_3)}{\alpha_1-\alpha_2})\\
\mathrm{Tr}(\gamma_1v_{t}f(\alpha_2))+\mathrm{Tr}(\frac{\gamma_1f(\alpha_1)}{\alpha_1-\alpha_2})+\mathrm{Tr}(\frac{\gamma_1f(\alpha_3)}{\alpha_1-\alpha_2})\\
\mathrm{Tr}(\gamma_2w_{t}f(\alpha_3))+\mathrm{Tr}(\frac{\gamma_2f(\alpha_1)}{\alpha_1-\alpha_2})+\mathrm{Tr}(\frac{\gamma_2f(\alpha_2)}{\alpha_1-\alpha_2})
\end{pmatrix}\\
&=&\begin{pmatrix}
2&1&1\\
1&2&1\\
1&1&2
\end{pmatrix}\begin{pmatrix}
\mathrm{Tr}(\frac{f(\alpha_1)}{\alpha_1-\alpha_2})\\
\mathrm{Tr}(\frac{f(\alpha_2)}{\alpha_1-\alpha_2})\\
\mathrm{Tr}(\frac{f(\alpha_3)}{\alpha_1-\alpha_2})
\end{pmatrix}.
\end{eqnarray*}
Since the coefficient matrix is invertible, the three nodes can directly compute the $t$-th independent trace for recovery.
\subsection{$l=t-2$}\label{sub2}
A three-round repair scheme for $t>3$ is designed here. When $l=t-2$, $K_{1,2}$, $K_{2,3}$, $K_{1,3}$ are distinct.
We can choose $u_{t-1}\!\in\! K_{1,2}\!\setminus\! K_{2,3}$ with $\mathrm{Tr}(u_{t-1}(\alpha_2-\alpha_3))\!=\!1$. Similarly, choose $v_{t-1}\!\in\! K_{2,3}\!\setminus\! K_{1,3}$ with $\mathrm{Tr}(v_{t-1}(\alpha_3-\alpha_1))\!=\!1$  and $w_{t-1}\!\in\! K_{1,3}\!\setminus\! K_{1,2}$ with $\mathrm{Tr}(w_{t-1}(\alpha_1-\alpha_2))\!=\!1$.
Then we have
$$\begin{cases}
\mathrm{Tr}(u_{t-1}\alpha_1)=\mathrm{Tr}(u_{t-1}\alpha_2)=1+\mathrm{Tr}(u_{t-1}\alpha_3) \\
1+\mathrm{Tr}(v_{t-1}\alpha_1)=\mathrm{Tr}(v_{t-1}\alpha_2)=\mathrm{Tr}(v_{t-1}\alpha_3)\\
\mathrm{Tr}(w_{t-1}\alpha_1)=1+\mathrm{Tr}(w_{t-1}\alpha_2)=\mathrm{Tr}(w_{t-1}\alpha_3)
\end{cases}\;. $$
Moreover, we can set $u_t\!=\!w_{t-1}$ since $w_{t-1}\!\notin\! K_{1,2}$. Similarly, set $v_t=u_{t-1}$ and $w_t=v_{t-1}$.

The Phase 1 is the same with that of the case $l=t-1$, except that the data obtained here is a little different due to different selections of the bases $U',V',W'$. The details are illustrated in Table \ref{t3}.
\begin{table}[ht]
	\centering
\resizebox{\columnwidth}{!}{
\begin{tabular}{|c|c|c|}
	\hline & \makecell[l]{download from $f(\alpha),$\\ $\alpha\in A_{\!-\!1,\!-\!2,\!-\!3}$} & obtain\\\hline
node 1 & $\mathrm{Tr}(\frac{f(\alpha)}{\alpha-\alpha_1})$ & $\begin{cases}
 \mathrm{Tr}(u_jf(\alpha_1)), ~~~~~j\in[t-2] \\
 \mathrm{Tr}(u_{t-1}f(\alpha_1))\!-\!\mathrm{Tr}(\frac{f(\alpha_3)}{\alpha_3-\alpha_1})\\
 \mathrm{Tr}(w_{t-1}f(\alpha_1))\!+\!\mathrm{Tr}(\frac{f(\alpha_2)}{\alpha_1-\alpha_2})
\end{cases}$\\\hline
node 2 & $\mathrm{Tr}(\frac{\gamma_1 f(\alpha)}{\alpha-\alpha_2})$ &
$\begin{cases}
 \mathrm{Tr}(\gamma_1v_jf(\alpha_2)), ~~~~~ j\in[t-2] \\
 \mathrm{Tr}(\gamma_1v_{t-1}f(\alpha_2))\!-\!\mathrm{Tr}(\frac{\gamma_1f(\alpha_1)}{\alpha_1-\alpha_2})\\
 \mathrm{Tr}(\gamma_1u_{t-1}f(\alpha_2))\!+\!\mathrm{Tr}(\frac{\gamma_1f(\alpha_3)}{\alpha_2-\alpha_3})
\end{cases}$\\\hline
node 3&$\mathrm{Tr}(\frac{\gamma_2 f(\alpha)}{\alpha-\alpha_3})$&$\begin{cases}
 \mathrm{Tr}(\gamma_2w_jf(\alpha_3)),  ~~~~~ j\in[t-2] \\
 \mathrm{Tr}(\gamma_2w_{t-1}f(\alpha_3))\!-\!\mathrm{Tr}(\frac{\gamma_2f(\alpha_2)}{\alpha_2-\alpha_3})\\
 \mathrm{Tr}(\gamma_2v_{t-1}f(\alpha_3))\!+\!\mathrm{Tr}(\frac{\gamma_2f(\alpha_1)}{\alpha_3-\alpha_1})
\end{cases}$\\\hline
\end{tabular}}
\caption{Phase 1 of the cooperative repair for three erasures ($l=t-2$).}\label{t3}
\end{table}

Next we specify the choice of  $\gamma_1$ and $\gamma_2$ to make sure the recovery can be realized by communication in Phase 2. First choose $\gamma_2\in F^*$ such that $\frac{\gamma_2^{-1}}{\alpha_3-\alpha_1}\in K_{1,2,3}$. Then
choose $\gamma_1\in F^*$ such that $\gamma_1^{-1}\gamma_2\in \gamma_2(\alpha_1-\alpha_2)K_{1,2,3}\cap K$. We can do this because dim$_{B}\big(\gamma_2(\alpha_1-\alpha_2)K_{1,2,3}\cap K\big)\geq t-3\geq1$ by Lemma \ref{lemma0} and the assumption that $t>3$.
Then phase 2 proceeds in three rounds as displayed in Fig. \ref{fig:1}.
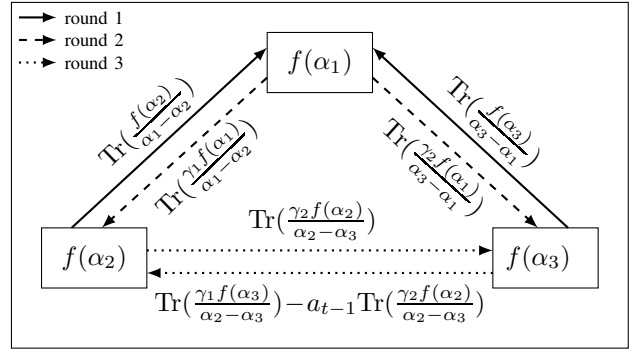
\begin{figure}[ht]
\begin{center}
\begin{tikzpicture}
\draw (0,0) rectangle (8.2,4.6);
\draw (0.4,0.8) rectangle (1.8,1.6);
\node at (1.1,1.2) {\bm$f(\alpha_2)$};
\draw (6.4,0.8) rectangle (7.8,1.6);
\node at (7,1.2) {\bm$f(\alpha_3)$};
\draw (3.4,3.4) rectangle (4.8,4.2);
\node at (4.1,3.8) {\bm$f(\alpha_1)$};

\draw[thick,->,>=latex] (0.8,1.6) -- (3.4,4);
\draw[thick,->,>=latex,dashed] (3.4,3.6) -- (1.2,1.6);
\draw[thick,->,>=latex] (7.4,1.6) -- (4.8,4);
\draw[thick,->,>=latex,dashed] (4.8,3.6) -- (7,1.6);
\draw[thick,->,>=latex,dotted] (1.8,1.3) -- (6.4,1.3);
\draw[thick,->,>=latex,dotted] (6.4,1) -- (1.8,1);

\node(a) at (1.8,3) [rotate=44] {$\mathrm{Tr}(\frac{f(\alpha_2)}{\alpha_1-\alpha_2})$};
\node(b) at (2.6,2.4) [rotate=44] {$\mathrm{Tr}(\frac{\gamma_1f(\alpha_1)}{\alpha_1-\alpha_2})$};

\node(c) at (6.4,3) [rotate=-42] {$\mathrm{Tr}(\frac{f(\alpha_3)}{\alpha_3-\alpha_1})$};
\node(d) at (5.6,2.4) [rotate=-42] {$\mathrm{Tr}(\frac{\gamma_2f(\alpha_1)}{\alpha_3-\alpha_1})$};

\node(e) at (4,1.7)  {$\mathrm{Tr}(\frac{\gamma_2f(\alpha_2)}{\alpha_2-\alpha_3})$};
\node(f) at (4.1,0.6){$\mathrm{Tr}(\frac{\gamma_1f(\alpha_3)}{\alpha_2-\alpha_3})\!-\!a_{t-1}\mathrm{Tr}(\frac{\gamma_2f(\alpha_2)}{\alpha_2-\alpha_3})$};

\draw[thick,->,>=latex] (0.1,4.4) -- (0.6,4.4);
\node at (1.1,4.4) {\scriptsize round 1};
\draw[thick,->,>=latex,dashed] (0.1,4.1) -- (0.6,4.1);
\node at (1.1,4.1) {\scriptsize round 2};
\draw[thick,->,>=latex,dotted] (0.1,3.8) -- (0.6,3.8);
\node at (1.1,3.8) {\scriptsize round 3};
\end{tikzpicture}
\end{center}
\caption{The exchange phase ($a_{t-1}$ comes from $\frac{\gamma_1}{\alpha_2-\alpha_3}=\sum_{i=1}^{t}a_i\gamma_2w_i$).}\label{fig:1}
\end{figure}

Specifically, round 1 can be achieved because $\frac{1}{\alpha_1-\alpha_2}\!\in\!\gamma_1K_{1,2,3}$ and $\frac{1}{\alpha_3-\alpha_1}\in\gamma_2K_{1,2,3}$ by the choice of $\gamma_1,\gamma_2$. Then $f(\alpha_1)$ can be recovered after round 1. As a result, round 2 can proceed. After round 2, node 2 recovers $\mathrm{Tr}(\gamma_1v_{t-1}f(\alpha_2))$ and node 3 recovers $\mathrm{Tr}(\gamma_2v_{t-1}f(\alpha_3))$. Then the two nodes exchange one subsymbol with each other in round 3. Actually, the round 3 is reduced to the Phase 2 when repairing two erasures. Since $\gamma_1^{-1}\gamma_2\!\in\! K$, a similar computation as that presented in Section \ref{sec3} is involved here for the final recovery.

\begin{remark}
In \cite{Dau18}, when $l=t-1$, a one-round repair scheme is designed restricting to ${\rm char}(F)\mid t$; when $l=t-2$, a three-round repair scheme is designed restricting to ${\rm char}(F)\mid t$ and $\{\frac{\alpha_1-\alpha_3}{\alpha_1-\alpha_2},\frac{\alpha_2-\alpha_1}{\alpha_2-\alpha_3},\frac{\alpha_3-\alpha_2}{\alpha_3-\alpha_1}\}\cap K\neq\emptyset$.
As in the repair of two erasures, we modify the check polynomials $q_i(x), r_i(x), i\in[t]$ by multiplying $\gamma_1$ and $\gamma_2$ respectively. Furthermore, we remove all the restrictions required in  \cite{Dau18} by properly choosing $\gamma_1, \gamma_2$.

\end{remark}
\section{Conclusions}\label{sec5}
We give an improved cooperative repair scheme for Reed-Solomon codes with two or three erasures, removing all the restrictions required in Dau et al.'s work \cite{Dau18}. An interesting problem in the future is to establish a lower bound on the repair bandwidth for cooperative repair of Reed-Solomon codes.

\end{document}